\documentclass{IEEEtran}
\usepackage{cite}
\usepackage{amsmath,amssymb,amsfonts,amsthm}
\usepackage{algorithmic}
\usepackage{graphicx}
\usepackage{textcomp}
\usepackage{algorithm}
\usepackage{algorithmic}
\usepackage{enumitem}
\usepackage{tabularx}
\makeatletter
\def\blfootnote{\xdef\@thefnmark{}\@footnotetext}
\makeatother

\newtheorem{example}{Example}
\newtheorem{theorem}{Theorem}
\newtheorem{lemma}[theorem]{Lemma}
\newtheorem{corollary}[theorem]{Corollary}
\newtheorem{remark}[theorem]{Remark}

\newtheorem{observation}[theorem]{Observation}
\newcommand{\Mod}[1]{\ (\mathrm{mod}\ #1)}
\newcolumntype{L}[1]{>{\raggedright\arraybackslash}p{#1}}
\newcolumntype{C}[1]{>{\centering\arraybackslash}p{#1}}
\newcolumntype{R}[1]{>{\raggedleft\arraybackslash}p{#1}}

\def\BibTeX{{\rm B\kern-.05em{\sc i\kern-.025em b}\kern-.08em
    T\kern-.1667em\lower.7ex\hbox{E}\kern-.125emX}}
\begin{document}
\title{Efficient and Secure Substitution Box and Random Number Generators Over Mordell Elliptic Curves}
\author{Ikram Ullah, Naveed Ahmed Azam\footnote, and Umar Hayat
\thanks{
}
\\
\thanks{
 Corresponding author}
\thanks{N. A. Azam is a research fellow at the Department of Applied Mathematics and Physics, Graduate School of Informatics, Kyoto University,
Japan. This research is partially funded through JSPS KAKENHI Grant Number 18J23484. (e-mail: azam@amp.i.kyoto-u.ac.jp).}
\thanks{
}
}

{}

\maketitle

\begin{abstract}
\textnormal{Elliptic curve cryptography has received great attention in recent years due to its high resistance against modern cryptanalysis.
The aim of this article is to present efficient generators to generate substitution boxes (S-boxes) and pseudo random numbers which are essential for many well-known cryptosystems.
These generators are based on a special class of ordered Mordell elliptic curves.
Rigorous analyses are performed to test the security strength of the proposed generators.
For a given prime, the experimental results reveal that the proposed generators are capable of generating a large number of distinct, mutually uncorrelated, cryptographically strong S-boxes and sequences of random numbers in low time and space complexity.
Furthermore, it is evident from the comparison that the proposed schemes can efficiently generate secure S-boxes and random numbers as compared to some of the well-known existing schemes over different mathematical structures.}
\end{abstract}

\section{Introduction}\label{Intro}
\label{sec:introduction}
Recent advancements in the field of communication systems and computational methods necessitate improvements in the traditional cryptosystems. Substitution box (S-box) and pseudo random number generator (PRNG) play an important role in many cryptosystems such as Data Encryption Standard (DES)~\cite{DES}, Advanced Encryption Standard (AES)~\cite{AES}, Twofish security system~\cite{F2}, Blowfish cryptosystem~\cite{BF},  International Data Encryption Algorithm (IDEA)~\cite{IDEA} and the cryptosystems developed in~\cite{umar2, NJia, XC, SS}.
It has been pointed out by many researchers that the security of a cryptosystem can be improved by using dynamic S-boxes instead of a single static S-box, see for example~\cite{RB, SK, MM, AS, MG, KK}.
This fact necessitates the development of new S-box generators which can generate a large number of distinct and mutually uncorrelated S-boxes with high cryptographic properties in low time and space complexity~\cite{Azam}.

Many researchers have proposed improved S-box generators and PRNGs to enhance the security of data against modern cryptanalysis.
These improvements are mainly based on finite field arithmetic and chaotic systems.
Khan and Azam~\cite{MN1, MN2} developed two different methods to generate 256 cryptographically strong S-boxes by using Gray codes, and affine mapping.
Jakimoski and Kocarev~\cite{Jaki} used chaotic maps to develop a four-step method for the generation of an S-box. $\rm \ddot{O}$zkaynak and $\rm \ddot{O}$zer~\cite{Ozkaynak2} introduced a new method based on a chaotic system to develop secure S-boxes. Unlike the traditional use of chaotic maps, Wang et al.~\cite{YW} proposed an efficient algorithm to construct  S-boxes using gnetic algorithm and chaotic maps. Yin et al.~\cite{Yin} proposed an S-box design technique using iteration of the chaotic maps. Tang and Liao~\cite{Tang} constructed S-boxes based on an iterating discretized chaotic map.
Lambi\'c~\cite{Lambic} used a special type of discrete chaotic map to obtain bijective S-boxes. \"Ozkaynak et al.~\cite{Ozkaynak} proposed a new S-box based on a fractional order chaotic Chen system.
 Zhang et al.~\cite{IChing} used I-Ching operators for the construction of highly non-linear S-boxes, and the proposed approach is very efficient.

Similarly, chaotic systems are used to generate pseudo random numbers (PRNs), see for example~\cite{MF, VP, CG, TS, ZF}. Francois et al.~\cite{MF} presented a PRNG based on chaotic maps to construct multiple  key sequences. Patidar and Sud~\cite{VP} designed a PRNG with optimal cryptographic properties using chaotic logistic maps. Guyeux et al.~\cite{CG} developed a chaotic PRNG with the properties of topological chaos which offers sufficient security for cryptographic purposes. Stojanovski and Kocarev~\cite{TS} analyzed a PRNG based on a piecewise linear one dimensional chaotic map. Fan et al.~\cite{ZF} proposed a PRNG using generalized Henon map, and a novel technique is used to improve the characteristics of the proposed sequences.

It has been pointed out by Jia et al.~\cite{NJia}
that the PRNs generated by a chaotic system can have small period due to the hardware computation issues and revealed that elliptic curve (EC) has high security than chaotic system.
However, the computation over ECs is usually performed by group law which is computationally inefficient.
Hayat and Azam~\cite{umar2}
proposed an efficient S-box generator and a PRNG based on ECs by using a total order as an alternative to group law.
This S-box generator is efficient than the other methods over ECs, however their time and space complexity are $\mathcal{O}(p^2)$ and $\mathcal{O}(p)$, respectively, where $p$ is the prime of the underlying EC.
Furthermore the S-box generator does not guarantee the generation of an S-box.
The PRNG proposed by Hayat and Azam~\cite{umar2} also takes $\mathcal{O}(p^2)$ and $\mathcal{O}(p)$ time and space, respectively, to generate a sequence of pseudo random numbers (SPRNs) of size $m \leq p$.
%
%
Azam et al.~\cite{Azam}
proposed an improved S-box generation method to generate bijective S-boxes by using ordered Mordell elliptic curves (MECs).
The main advantage of this method is that its time and space complexity are $\mathcal{O}(mp)$ and $\mathcal{O}(m)$, respectively, where $m$ is the size of an S-box.
Azam et al.~\cite{ikram}
proposed another S-box generator to generate $m \times n$, where $m \leq n$ injective S-boxes which can generate a large number of distinct and mutually uncorrelated S-boxes by using the concept of isomorphism on ECs.
The time and space complexity of this method are $\mathcal{O}(2^np)$ and $\mathcal{O}(2^n)$, where $n \leq p$ and is the size of co-domain of the resultant S-box.
A common draw back of these S-box generators is that the cryptographic properties of their optimal S-boxes are far from the theoretical optimal values.

The aim of this paper is to propose an efficient S-box generator and a PRNG based on an ordered MEC to generate a large number of distinct, mutually uncorrelated S-boxes and PRNs with optimal cryptographic properties in low time and space complexity to overcome the above mentioned drawbacks.
The rest of the paper is organized as follows:
In Section~\ref{Prel} basic definitions are discussed. The proposed S-box generator is described in Section~\ref{Cons}. Section~\ref{Anal} consists of security analysis and comparison of the S-box generator.
The proposed algorithm for generating PRNs and some general results are given in Section~\ref{RNG}.
The proposed SPRNs are analyzed in Section~\ref{RA}, while Section~\ref{Con} concludes the whole paper.

\section{Preliminaries}\label{Prel}
Throughout this paper, we denote a finite set $\{0, 1, \ldots, m-1\}$ simply by $[0, m-1]$.
A finite field over a prime number $p$ is the set $[0, p-1]$ denoted by $\mathbb{F}_{p}$ with binary operations addition and multiplication under modulo $p$.
A non-zero integer $\alpha \in \mathbb{F}_{p}$ is said to be {\it quadratic residue} (QR) if there exists an integer $\beta \in \mathbb{F}_p$ such that $\alpha\equiv \beta^2 \pmod p$.
A non-zero integer in $\mathbb{F}_{p}$ which is not QR is said to be {\it quadratic non-residue} (QNR).

For a prime $p$, non-negative $a\in \mathbb{F}_{p}$ and positive $b\in \mathbb{F}_{p}$, the EC $E_{p,a,b}$ over a finite field $\mathbb{F}_{p}$ is defined to be the collection of identity element $\infty$ and ordered pairs $(x, y) \in \mathbb{F}_p \times \mathbb{F}_p$ such that \[ y^{2}\equiv x^{3}+ax+b \pmod p.\]
In this setting, we call $p,a$ and $b$ the {\it parameters} of $E_{p,a,b}$.
The number $\#E_{p,a,b}$ of all such points can be calculated using Hasse's
theorem~\cite{Schoof} \[|\#E_{p,a,b}-p-1|\leq 2\sqrt{p}.\]
Two ECs $E_{p,a,b}$ and $E_{p,a',b'}$ over $\mathbb{F}_{p}$ are isomorphic if and only if there exists a non-zero integer $t\in \mathbb{F}_{p}$ such that $a't^4\equiv a \pmod p$ and $b't^6\equiv b \pmod p$.
In this case, $t$ is called {\it isomorphism parameter} between the ECs $E_{p,a,b}$ and $E_{p,a',b'}$.
For an isomorphism parameter $t$, each point $(x,y)\in E_{p,a,b}$ is mapped on $(t^2x,t^3y)\in E_{p,a',b'}$.
Note that an isomorphism is an equivalence relation on all ECs over $\mathbb{F}_{p}$, and therefore all ECs can be divided into equivalence classes~\cite{ikram}. For the sake of simplicity we represent an arbitrary class by $\mathcal{C}_i$ and assume that the class $\mathcal{C}_1$ contains the EC $E_{p,0,1}$.
A non-negative integer $b\in \mathbb{F}_{p}$ such that $E_{p,a,b}\in \mathcal{C}_i$ is called {\it representative} $R(p, \mathcal{C}_i)$ of the class $\mathcal{C}_i$.
Clearly, it holds that $R(p,\mathcal{C}_1) = 1$.

An EC $E_{p,a,b}$ with $a = 0$ is said to be a {\it Mordel elliptic curve}.
The following theorem is from~\cite[6.6 (c), p.~188]{Schoof}.
\begin{theorem}\label{mordel}
Let $p>3$ be a prime such that $p\equiv 2\pmod{3} $.
For each non-zero $b\in \mathbb{F}_{p}$, the MEC $E_{p,0,b}$ has exactly $p+1$
distinct points, and has each integer from $[0,p-1]$ exactly once as a $y$%
-coordinate.
\end{theorem}
\noindent
Furthermore, by~\cite[Lemma 1]{ikram} it follows that there are only two classes of MECs when $p\equiv 2 \pmod 3$.
Henceforth, we denote an MEC $E_{p,0, b}$ by simply $E_{p, b}$ and the term MEC stands for an MEC $E_{p, b}$ such that $p\equiv 2 \pmod 3$.

For a subset $A$ of $[0, p-1]$ and an ordered MEC $(E_{p, b}, \prec)$, we define a total order $\prec^*$ on $A$ w.r.t. the ordered MEC such that for any two elements $a_1, a_2 \in A$ it holds that $a_1 \prec^* a_2$ if and only if $(x_1, a_1) \prec (x_2, a_2)$.
For any two non-negative integers $p$ and $m$ such that $1 \leq m \leq p$, we define an {\em $(m, p)$-complete set} to be a set $Q$ of size $m$ such that for each element $ q \in Q$, it holds that $0\leq q \leq p-1$, 
and no two elements of $Q$ are congruent with each other under modulo $m$, i.e., for each $q,q' \in Q$, it holds that $(q \not\equiv q') \pmod{m}$.
We denote an ordered set $A$ with a total order $\prec$ by an ordered pair $(A, \prec)$.
Let $(A, \prec)$ be an ordered set, for any two elements $a,a' \in A$ such that $a \prec a'$, we read as $a$ is smaller than $a'$ or $a'$ is greater than $a$ w.r.t. the order $\prec$.
For simplicity, we represent the elements of $(A, \prec)$ in the form of a non-decreasing sequence and $a_i$ denotes the $i$-th element of the ordered set in its sequence representation.
For an ordered MEC $(E_{p, b},\prec)$ and an $(m,p)$-complete set $Y$, we define the \textit{ordered $(m,p)$-complete set} $Y^*$ with ordering $\tilde\prec$ due to $Y$ and $\prec$ such that for any two element $y_1, y_2 \in Y^*$ with $y_1' \equiv y_1 \pmod{m}$ and $y_2' \equiv y_2 \pmod m$, where $y_1', y_2' \in Y$, it holds that $y_1 \tilde\prec y_2$ if and only if $(x_1, y_1') \prec (x_2, y_2')$. 

For a given MEC $E_{p, b}$, Azam et al.~\cite{Azam} defined three typical type of orderings \textit{natural} $N$, \textit{diffusion} $D$ and \textit{modulo diffusion} $M$ ordering  based on the coordinates of the points on $E_{p, b}$ as\\
$(x_{1},y_{1})N(x_{2},y_{2}) \Leftrightarrow
\begin{cases}
{\rm either \ } x_{1}<x_{2}, {\rm \ or} \\[5pt]
x_{1}=x_{2} {\rm \ and \ } y_{1}<y_{2}; %
\end{cases}
\\ [5pt]
(x_{1},y_{1})D(x_{2},y_{2}) \Leftrightarrow
\begin{cases}
{\rm either \ } x_{1}+y_{1}<x_{2}+y_{2}, {\rm \ or} \\[5pt]
x_{1}+y_{1}=x_{2}+y_{2} {\rm \ and \ } x_{1}<x_{2};
\end{cases}
\\[5pt]
(x_{1},y_{1})M(x_{2},y_{2}) \Leftrightarrow \\[5pt]
\begin{cases}
{\rm either \ } (x_{1}+y_{1}<x_{2}+y_{2})\pmod p, {\rm \ or} \\[5pt]
x_{1}+y_{1}\equiv x_{2}+y_{2}\pmod p
{\rm \ and \ } x_{1}<x_{2}.
\end{cases}
$
%
%
\section{The Proposed S-box Construction Scheme}\label{Cons}
For an ordered MEC $(E_{p,b}, \prec)$, an $(m,p)$-complete set $Y$ and a non-negative integer $k \leq m -1$, we define an \textit{$(m,p)$-complete S-box} $\sigma(p,b,\prec,Y, k)$ due to $(E_{p,b}, \prec)$, $Y$ and $k$ to be a mapping from $[0, m-1]$ to $(Y^*, \tilde\prec)$ such that $\sigma(p,b,\prec,$ $Y,k)(i) = y_{(i +k)\pmod{m}}\pmod{m}$, where $y_{(i +k)\pmod{m}}$ is the $(i +k)\pmod{m}$-th element of the ordered $(m,p)$-complete set $(Y^*, \tilde\prec)$ in its sequence representation.
\begin{lemma}\label{bijective}
For any ordered MEC $(E_{p,b}, \prec)$, an $(m,p)$-complete set $Y$ and a non-negative integer $k \leq m -1$, the $(m,p)$-complete S-box $\sigma(p,b,\prec,Y, k)$ is a bijection.
\end{lemma}
\begin{proof}Suppose on contrary that there exist $i,j \in [0, m-1]$ such that $\sigma_{p,b,\prec,Y,k}(i) = \sigma_{p,b,\prec,Y,k}(j)$.
This implies that $y_i' = y_j'$, where $y_i' \equiv y_i\pmod{m}$ and $y_j' \equiv y_j\pmod{m}$ and $y_i', y_j' \in Y$.
This leads to a contradiction to the fact that $Y$ is an $(m,p)$-complete set.
Thus, $\sigma(p,b,\prec,Y, k)$ is a one-one mapping on the finite sets of same order, and hence it is a bijection.
\end{proof}

For prime $p = 52511$ and $m = 256$, an  $(256,52511)$-complete subset $Y$ of the MEC is given in Table~\ref{FT}, while the $(256,52511)$-complete S-box $\sigma(52511, 1, N, Y, 0)$ due to the ordered MEC $(E_{52511,1}, N)$ is presented in Table~\ref{SN} in hexagonal format.
Each entry of Table~\ref{SN} is obtained from the corresponding entry of Table~\ref{FT} by applying modulo $256$ operator.
\begin{table*}[htb]
\caption{The  $(256,52511)$-complete set $Y$}
\label{FT}%
\centering
\scalebox{0.99}
{\begin{tabular}{llllllllllllllll}
\hline
 A792 & 4A5C & 9AF5 & 01C5 & 421  & 814D & B3A2 & 5CA3 & 834B & 9F90 & 1C7D & BF6A & 0A11 & 7A9D & 9E91 & 6135 \\
1D8D & 9425 & 3F36 & 7954 & 1E1E & 5B47 & 1420 & 71CA & 8089 & 80C4 & 3150 & 12EF & 36C3 & BEA7 & 6170 & 2256 \\
298A & 005E & 8032 & 0F00 & 270A & 51D0 & 421C & 942  & 6BDF & 2848 & 87FC & 4418 & 2BFB & 121B & 6F2D & 11CC \\
886  & 8E53 & 6BD2 & AC14 & B65B & 062B & 37F1 & B627 & 47D4 & 59A6 & 2878 & 7D76 & 76CB & 7005 & 0CBE & 8F8F \\
609E & 7A83 & 61F4 & 23C0 & 3AC2 & 3502 & BC40 & 88DE & 3645 & 2EEC & B8B3 & BBD9 & 84D5 & 165F & C061 & 0BE8 \\
AE34 & 6431 & 906B & 15E4 & BE74 & 5423 & 10AA & 4D75 & B037 & 556F & 6F99 & 242E & 31AD & C9F9 & A679 & 3F82 \\
749A & 7F55 & 9267 & AF29 & 33BC & 1A0E & 270D & 2312 & 7857 & B730 & 5C17 & AAA4 & 7DF7 & 698B & 7FDB & 66AC \\
A203 & B46D & 7DE9 & 7E80 & 72B2 & 97E1 & 70D1 & 18D8 & 76ED & 4677 & 7A4E & 7F3F & 96B8 & 8A94 & 91D3 & 8295 \\
6E0F & 7A0B & 221A & 11C7 & A7A1 & 1563 & 33BB & 15EA & 62BA & 0EB9 & 8041 & 6998 & C260 & 127C & 0B2F & 38AE \\
7626 & 12A8 & 50D6 & B0CE & 67CD & 766C & 22BD & 109F & 4E4C & CBF2 & 5CA5 & 2528 & 1964 & 4724 & CAAF & 966  \\
A587 & AE01 & 5584 & 0A3D & 3859 & 7504 & 063E & 5251 & 767B & 0AFF & 50E7 & 7765 & 2688 & BC58 & 972A & 0EE6 \\
295A & 0BB6 & 4B43 & 5906 & 476E & C5C9 & 20A9 & 45AB & C57A & 1D07 & 694A & B57F & 0D15 & 1CBF & ABFA & C3B5 \\
2096 & B138 & A671 & A262 & BAF3 & 4CB1 & 054F & C5DD & 9B85 & C144 & BBDC & 7969 & C85D & 91C6 & 0A49 & 9DE2 \\
C6DA & 278C & 1C13 & 29D7 & 708E & 827E & 0FC8 & 4FEB & 4BEE & 1F97 & 20FE & 26E0 & 0E93 & 4E9C & A2E5 & 841D \\
ADF0 & B273 & A6E3 & 440C & AB08 & 3952 & 103A & A472 & C42C & 36CF & 9768 & 6809 & 0E22 & C439 & 291F & AEFD \\
A7C1 & C23B & 2FF6 & A046 & 3BB4 & ACA0 & 5A9B & 95F8 & 7919 & 4381 & A9B0 & 7110 & 7433 & 1816 & 39B7 & 1A3C\\ \hline
\end{tabular}}
\end{table*}

\begin{table}[H]
\caption{The  proposed $(256,52511)$-complete S-box $\sigma(52511, 1, N, Y, 0)$}
\label{SN}%
\centering
\resizebox{\columnwidth}{!}
{\begin{tabular}{cccccccccccccccc}
\hline
92 & 5C & F5 & C5 & 21 & 4D & A2 & A3 & 4B & 90 & 7D & 6A & 11 & 9D & 91 & 35 \\
8D & 25 & 36 & 54 & 1E & 47 & 20 & CA & 89 & C4 & 50 & EF & C3 & A7 & 70 & 56 \\
8A & 5E & 32 & 0  & 0A & D0 & 1C & 42 & DF & 48 & FC & 18 & FB & 1B & 2D & CC \\
86 & 53 & D2 & 14 & 5B & 2B & F1 & 27 & D4 & A6 & 78 & 76 & CB & 5  & BE & 8F \\
9E & 83 & F4 & C0 & C2 & 2  & 40 & DE & 45 & EC & B3 & D9 & D5 & 5F & 61 & E8 \\
34 & 31 & 6B & E4 & 74 & 23 & AA & 75 & 37 & 6F & 99 & 2E & AD & F9 & 79 & 82 \\
9A & 55 & 67 & 29 & BC & 0E & 0D & 12 & 57 & 30 & 17 & A4 & F7 & 8B & DB & AC \\
3  & 6D & E9 & 80 & B2 & E1 & D1 & D8 & ED & 77 & 4E & 3F & B8 & 94 & D3 & 95 \\
0F & 0B & 1A & C7 & A1 & 63 & BB & EA & BA & B9 & 41 & 98 & 60 & 7C & 2F & AE \\
26 & A8 & D6 & CE & CD & 6C & BD & 9F & 4C & F2 & A5 & 28 & 64 & 24 & AF & 66 \\
87 & 1  & 84 & 3D & 59 & 4  & 3E & 51 & 7B & FF & E7 & 65 & 88 & 58 & 2A & E6 \\
5A & B6 & 43 & 6  & 6E & C9 & A9 & AB & 7A & 7  & 4A & 7F & 15 & BF & FA & B5 \\
96 & 38 & 71 & 62 & F3 & B1 & 4F & DD & 85 & 44 & DC & 69 & 5D & C6 & 49 & E2 \\
DA & 8C & 13 & D7 & 8E & 7E & C8 & EB & EE & 97 & FE & E0 & 93 & 9C & E5 & 1D \\
F0 & 73 & E3 & 0C & 8  & 52 & 3A & 72 & 2C & CF & 68 & 9  & 22 & 39 & 1F & FD \\
C1 & 3B & F6 & 46 & B4 & A0 & 9B & F8 & 19 & 81 & B0 & 10 & 33 & 16 & B7 & 3C \\ \hline
\end{tabular}}
\end{table}


Next we present two efficient algorithms to compute an $(m,p)$-complete S-box.
The first algorithm is based on Theorem~\ref{mordel}.
\begin{algorithm}[H]
	\caption{{\bf Constructing the proposed S-box } } \label{algo_sbox_1}
\begin{algorithmic}[1]
\REQUIRE An ordered MEC $(E_{p,b}, \prec)$, an $(m,p)$-complete set~$Y$ and a non-negative integer $k\leq m-1$.
\ENSURE The proposed $(m,p)$-complete S-box $\sigma(p,b,\prec,$ $Y,k)$.

\STATE $A:=\emptyset$; /*A set containing the points of $E_{p,b}$ with $y$-coordinates from $Y$*/

\FORALL {$y \in Y$}
\STATE $B:= [0, p-1]$; $t:= $ {\tt No};
	\WHILE{$q=$ {\tt No}}
		\STATE $x \in B$;
		\IF{$(x^3 +b - y^2 \equiv 0) \pmod p$}
			\STATE $q:=$ {\tt Yes}; $A:= A \cup \{(x,y)\}$
		\ENDIF;
		\STATE $B:=B\setminus \{x\}$
	\ENDWHILE
\ENDFOR;
\STATE Sort $Y^* := [0, m-1]$ w.r.t. the element of $A$, i.e., sort $Y^*$ w.r.t. the ordering $\tilde\prec$;
\STATE $\pi := (\pi(0), \pi(1), \ldots, \pi(m-1))$;
\FORALL {integer $i \in [0, m-1]$}
	\STATE $\pi(i) := y_{(i+k)\Mod{m}}\pmod{m}$, where $y_{(i +k)\pmod{m}}$ is the $(i +k)\pmod{m}$-th element of the ordered $(m,p)$-complete set $(Y^*, \tilde\prec)$
\ENDFOR;
\STATE Output $\pi$ as the $(m,p)$-complete S-box $\sigma(p,b,\prec,Y, k)$.
	
\end{algorithmic}
\end{algorithm}
\begin{lemma}\label{imp_1}
For an ordered MEC $(E_{p,b}, \prec)$, an $(m,p)$-complete set $Y$ and a non-negative integer $k\leq m-1$, the $(m,p)$-complete S-box $\sigma(p,b,\prec,Y, k)$ can be computed in $\mathcal{O}(mp)$ time and $\mathcal{O}(m)$ space by using Algorithm~\ref{algo_sbox_1}.
\end{lemma}
\begin{proof}
In Algorithm~\ref{algo_sbox_1} there is for-loop of size $m$ over the elements of $Y$, which has a nested while-loop to compute the subset $A$ of the MEC such that the points in $A$ has $y$-coordinate in $Y$.
This step is necessary to compute the ordered $(m,p)$-complete set $(Y^*, \tilde\prec)$ due to $Y$ and $\prec$.
Note that the nested while-loop will iterate for at most $p$-times, since by Theorem~\ref{mordel}, for each $y \in [0, p-1]$ there is a unique $x \in [0, p-1]$ such that $(x^3 +b - y^2 \equiv 0) \pmod p$.
Thus, this for-loop and while-loop take $\mathcal{O}(mp)$ time in the worst case, while the sorting of $Y^*$ takes $\mathcal{O}(m\log m)$ time.
Finally, there is another independent for-loop of size $m$ to compute the sequence $\pi$ which takes $\mathcal{O}(m)$ time.
Thus, Algorithm~\ref{algo_sbox_1} takes $\mathcal{O}(mp) + \mathcal{O}(m\log m) + \mathcal{O}(m)$ time to execute in the worst case.
By using the fact that $mp > m \log p$, since $\log{p} < p$ and $m \leq p$ and by the property of $\mathcal{O}$ notation, the time complexity of Algorithm~\ref{algo_sbox_1} is $\mathcal{O}(mp)$.
Furthermore, Algorithm~\ref{algo_sbox_1} only stores sets of size $m$, and therefore its space complexity is $\mathcal{O}(m)$.
This completes the proof.
\end{proof}
Next we present another algorithm for the generation of $(m,p)$-complete S-boxes on a fixed MEC.
For this we prove the following results.

For a fixed ordered MEC $(E_{p,b}, \prec)$, a positive integer $m \leq p$ and an integer $0 \leq k \leq m - 1$, let Num$(E_{p,b}, \prec, m, k)$ denote the total number of $(m,p)$-complete S-boxes, possibly with repetition, generated due to the ordered MEC, $m$ and $k$.
\begin{lemma}\label{LL}
For a fixed ordered MEC $(E_{p,b}, \prec)$ and a positive integer $m \leq p$,  the total number of $(m,p)$-complete S-boxes, possibly with repetition, generated due to the MEC is equal to $m(q+1)^{r}q^{m-r}$, where $p=mq+r$, $0\leq r< m,$ and $0\leq k \leq m-1$.
\end{lemma}
\begin{proof}
For a fixed integer $0 \leq k \leq m-1$, it holds by the definition of $(m,p)$-complete S-box that the total number of $(m,p)$-complete S-boxes, possibly with repetition, generated due to the ordered MEC, $m$ and $k$
is equal to the number of distinct $(m,p)$-complete sets.
If $p = mq + r$, where $0 \leq r \leq m-1$, then there are $q+1$ (resp., $q$) integers $\ell$ (resp., $h$)
such that $\ell \pmod m \in [0, r-1]$ (resp., $h \pmod m \in [r, m-1]$).
Thus, to construct an $(m,p)$-complete set there are $q+1$ (resp., $q$) choices of an integers $a$ such that $a \pmod m \in [0, r-1]$ (resp., $[r, m-1]$).
This implies that there are $(q+1)^r q^{m-r}$ distinct $(m,p)$-complete sets.
Hence, the number of $(m, p)$-complete S-boxes due to the MEC is $m(q+1)^r q^{m-r}$, since $0 \leq k \leq m-1$.
\end{proof}
\begin{observation}\label{obs_set}
For any subset $F$ of an MEC $E_{p,b}$ there exists a unique subset $F'$ of either MEC $E_{p, R(p,\mathcal{C}_{1})}$ or $E_{p, R(p,\mathcal{C}_{2})}$ and a unique integer $t \in [1, (p-1)/2]$ such that for each $(x,y) \in F$ there exists a unique point $(x', y') \in F'$ for which it holds that $x\equiv t^2x'\pmod p$ and $y \equiv t^3y'\pmod p$.
\end{observation}
\noindent
It is important to mention that for each subset $F$ such that the set of $y$-coordinates of its points is an $(m,p)$-complete set, the set of $y$-coordinates of the points of $F'$ is not necessarily be an $(m,p)$-complete set.
This is explained in Example~\ref{E2}.
\begin{example}\label{E2}
Let $F$ be a subset of $E_{11,9}$ with an $(11,10)$-complete set $Y= \{0,1,2,3,4,5,6,7,8,9\}$ of $y$-coordinates, where $m = 10$.
Then for $t = 2$, there exists $F' \subset E_{11, 1}$ with $y$-coordinates from the set $Y'=\{0,1,2,3,5,6,7,8,9,10\} $ which is not an $(11,10)$-complete set.
\end{example}
By Observation~\ref{obs_set}, we can avoid the while-loop used in Algorithm~\ref{algo_sbox_1} to find $x$-coordinate for each element $y$ in an $(m,p)$-complete set $Y$.
\begin{algorithm}[H]
	\caption{{\bf Constructing the proposed S-box using the EC isomorphism }}\label{algo_sbox_2}
\begin{algorithmic}[1]
\REQUIRE An MEC $E_{p, R(p,\mathcal{C}_{i})}$, where $i \in \{1, 2\}$, multiplicative inverse $t^{-1}$ of $t$ in $\mathbb{F}_p$, where $t \in [1, (p-1)/2]$, a total order $\prec$ on the MEC $E_{p,t^{6}R(p,\mathcal{C}_{i})}$, an $(m,p)$-complete set $Y$ and an integer $k \leq m-1$.
\ENSURE The proposed $(m,p)$-complete S-box $\sigma(p, t^6R(p, \mathcal{C}_i), \prec, Y, k)$.

\STATE $A:=\emptyset$; /*A set containing the points of $E_{p,t^6R(p,\mathcal{C}_{i})}$ with $y$-coordinates from the set $Y$*/

\FORALL {$y \in Y$}
\STATE $y':= (t^{-1})^3y$;
\STATE Find $x \in [0, p-1]$ such that $(x, y) \in E_{p,R(p,\mathcal{C}_{i})}$;
\STATE $A:= A \cup \{(t^2x, y)\}$;
\ENDFOR;
\STATE Sort $Y^*:=[0, m-1]$ w.r.t. the element of $A$;
\STATE $\pi := (\pi(0), \pi(1), \ldots, \pi(m-1))$;
\FORALL {integer $i \in [0, m-1]$}
	\STATE $\pi(i) := y_{(i+k)\Mod{m}}\pmod{m}$, where $y_{(i +k)\pmod{m}}$ is the $(i +k)\pmod{m}$-th element of the ordered $(m,p)$-complete set $(Y^*, \tilde\prec)$
\ENDFOR;
\STATE Output $\pi$ as the $(m,p)$-complete S-box $\sigma(p, t^6R(p, \mathcal{C}_i), \prec,$ $ Y, k)$.
	
\end{algorithmic}
\end{algorithm}
\begin{lemma}
For an ordered MEC $(E_{p,b}, \prec)$, where $b= t^6R(p,\mathcal{C}_{i})$ for some $t \in [1, (p-1)/2]$ and $i \in \{1, 2\}$, an $(m,p)$-complete set $Y$ and a non-negative integer $k\leq m-1$, the $(m,p)$-complete S-box $\sigma(p, b, \prec, Y, k)$ can be computed in $O(m\log m)$ time and $\mathcal{O}(m)$ space by using Algorithm~\ref{algo_sbox_2}.
\end{lemma}
\begin{proof}
There is a for-loop over the set $Y$ of size $m$ for finding $x$-coordinate for each element $y \in Y$ over the MEC $E_{p,t^{6}R(p,\mathcal{C}_{i})}$.
Note that at line 4 of Algorithm~\ref{algo_sbox_2}, $x$ can be computed in constant time, i.e., $\mathcal{O}(1)$.
This is due to Theorem~\ref{mordel} the MEC $E_{p,b}$ has each element of $[0, p-1]$ uniquely as $y$-coordinate. 
Thus, the for-loop over $Y$ can be computed in $\mathcal{O}(m)$.
The remaining part of Algorithm~\ref{algo_sbox_2} takes $O(m\log m)$ time.
Hence, with the aid of the property of $\mathcal{O}$ notion,  Algorithm~\ref{algo_sbox_2} takes $O(m\log m)$ time.
Moreover, Algorithm~\ref{algo_sbox_2} stores only a set of size $m$, other than inputs, and therefore its space complexity is  $\mathcal{O}(m)$.
\end{proof}
Note that using Algorithm~\ref{algo_sbox_2} is practical, since Lemma~\ref{LL} implies that for a given ordered MEC $(E_{p,b}, \prec)$ we can generate a large number of $(m,p)$-complete S-boxes.
However, $E_{p, R(p,\mathcal{C}_{i})}$, where $i \in \{1, 2\}$, $R(p, \mathcal{C}_i)$ and $t^{-1}$ for $t \in [0, (p-1)/2]$ should be given as input for Algorithm~\ref{algo_sbox_2}.
We know that $R(p,\mathcal{C}_{1}) = 1$, now the next important question is how to find the representative $R(p,\mathcal{C}_{2})$ for the class $\mathcal{C}_{2}$ of MECs.
For this we prove the following results.
\begin{lemma}\label{x_zero}
An MEC $E_{p,b}$ is an element of the class $\mathcal{C}_1$ if and only if there exists an integer $y \in [1, p-1]$ such that $(0, y) \in E_{p,b}$.
\end{lemma}
\begin{proof}
Consider the MEC $E_{p, 1}$.
Then for $y = 1$ the equation $x^3 + 1 \equiv 1 \pmod p$ is satisfied by $x = 0$.
This implies that $(0, 1) \in E_{p, 1}$, and hence the required statement is true for the MEC $E_{p, 1}$.
Let $E_{p, b} \in \mathcal{C}_1$, where $b \in [2, p-1]$.
Then there exists an isomorphism parameter $t \in [1, (p-1)/2]$ between $E_{p, 1}$ and $E_{p, b}$ such that $(t^20, t^31) = (0, t^3) \in E_{p, b}$.
%
%
%
Hence, for each MEC $E_{p,b} \in \mathcal{C}_1$ there exists an integer $y \in [1, p-1]$ such that $(0, y) \in E_{p,b}$.

To prove the converse, suppose on contrary that there is an MEC $E_{p,b}$ with a point $(0,y)$ for some $y \in [1, p-1]$ and $E_{p,b} \notin \mathcal{C}_1$.
This implies that there does not exist an integer $t \in  [1, (p-1)/2]$ such that $b \equiv t^6 \pmod p$.
Thus, $b \not\equiv (t^3)^2 \pmod p$ for all $t \in   [1, (p-1)/2]$.
But it follows from $(0,y) \in E_{p,b}$ that $b \equiv y^2 \pmod p$ for some $y \in [1, (p-1)/2]$ which is a contradiction.
Hence $E_{p,b} \in \mathcal{C}_1$.
\end{proof}

\begin{lemma}\label{class_2_rep}
For a prime $p$, the representative $R(p,\mathcal{C}_2)$ of the class $\mathcal{C}_2$ is a QNR integer in the field $\mathbb{F}_{p}$.
\end{lemma}
\begin{proof}
Let $E_{p,b} \in \mathcal{C}_2$.
Suppose on contrary that $b$ is a quadratic integer in the field $\mathbb{F}_p,$ i.e., $b \equiv y^2 \pmod p$ for some integer $y \in [1, p-1]$.
It follows from the equation $x^3 + b \equiv y^2 \pmod p$ that $(0, y) \in E_{p,b}$.
By Lemma~\ref{x_zero}, it holds that $E_{p,b} \in \mathcal{C}_1$, which is a contradiction to our assumption.
So, $b$ is a QNR, and hence $R(p,\mathcal{C}_2)$ is a QNR.
\end{proof}
Euler's Criterion is a well-known method to test if a non-zero element of the field $\mathbb{F}_p$ is a QR or not.
We state this test in Lemma~\ref{EC}.
\begin{lemma}\label{EC}
\textnormal{\cite[p.~1797]{Sze}} An element $q\in\mathbb{F}_{p}$ is a QR  if and only if $q^{(p-1)/2}\equiv -1 \pmod p$.
\end{lemma}
\section{Security Analysis and Comparison}\label{Anal}
In this section, a detailed analysis of the proposed S-box is performed.
Most of the cryptosystems use $8 \times 8$ S-boxes and therefore, we use $8 \times 8$ $(256, 525211)$-complete S-box $\sigma(52511, 1, N, Y, 0)$ given in Table~\ref{SN} generated by the proposed method for experiments.
The cryptographic properties of the proposed S-box are also compared with some of the well-known S-boxes developed by different mathematical structures.
\subsection{Linear Attacks}
Linear attacks are used to exploit linear relationship between input and output bits.
A cryptographically strong S-box is the one which can strongly resist linear attacks. The resistance of an S-box against linear attacks is evaluated by well-known tests including non-linearity~\cite{Carlet}, linear approximation probability~\cite{Matsui} and algebraic complexity~\cite{Sakalli}.
For a bijective $n \times n$ S-box $S$, the non-linearity NL$(S)$, linear approximation probability LAP$(S)$ can be computed by Eqs.~(\ref{NL}) and~(\ref{LAP}), respectively, while its algebraic complexity AC$(S)$ is measured by the number of non-zero terms in its linearized algebraic expression~\cite{Lidl}.
\begin{equation}\label{NL}
{\rm NL}(S)=\min_{\alpha ,\beta ,\lambda }\{x\in \mathbb{F}_{2}^{n}:\alpha \cdot
S(x)\neq \beta \cdot x\oplus \lambda \},
\end{equation}
\begin{equation}\label{LAP}
\begin{split}
&{\rm LAP}(S)=\frac{1}{2^{n}}\Big\{\max_{\alpha ,\beta }\big\{|\#\{x\in \mathbb{F}_{2}^{n} \mid \\
 & \hspace{4cm}x \cdot \alpha=S(x) \cdot \beta \}-2^{n-1}|\big \}\Big \},
\end{split}
\end{equation}
where $\alpha \in \mathbb{F}_{2}^{n}${\normalsize , }$\lambda \in \mathbb{F}_{2}$%
{\normalsize , }$\beta \in \mathbb{F}_{2}^{n}\backslash \{0\}${\normalsize \ and
\textquotedblleft }$\cdot ${\normalsize \textquotedblright\ represents the inner
product over }$\mathbb{F}_{2}.$

An S-box $S$ is said to be highly resistive against linear attacks if it has NL close to $2^{n-1} - 2^{(n/2)-1}$, low LAP and AC close to $2^n-1$.

The experimental results of NL, LAP and AC of the proposed S-box $\sigma(52511, 1, N, Y, 0)$ and some of the well-known S-boxes are given in Table~\ref{allanalysis}.
Note that the proposed S-box has NL, LAP and AC close to the optimal values.
The $\rm NL$ of $\sigma(52511, 1, N, Y, 0)$ is greater than that of the S-boxes in~\cite{Asif, ikram, umar2, Azam, YW, GT, Jaki, Ozkaynak2,Bhattacharya,IChing,YWang,AGautam,Shi} and equal to that of~\cite{AES}.
The $\rm LAP$ of $\sigma(52511, 1, N, Y, 0)$ is less than that of the S-boxes in~\cite{Asif, ikram, umar2, Azam, YW, GT, Jaki, Ozkaynak2,Bhattacharya,YWang,IChing,AGautam,Shi},  and the AC of $\sigma(52511, 1, N, Y, 0)$ attains the optimal value, which is $255$.
Thus the proposed method is capable of generating S-boxes with optimal resistance against linear attacks as compared to some of the existing well-known S-boxes.
\subsection{Differential Attack}
In this attack, cryptanalysts try to approximate the original message by observing a particular difference in output bits for a given input bits difference.
The strength of an $n \times n$ S-box $S$ can be measured by calculating its differential approximation probability DAP$(S)$ using Eq.~(\ref{DAP}). 
\begin{equation}\label{DAP}
\begin{split}
&{\rm DAP}(S)=\frac{1}{2^{n}}\Big\{\max_{\Delta x,\Delta y}\big\{\#\{x\in \mathbb{F}_{2}^{n} \mid\\
& \hspace{3cm} S(x\oplus \Delta x)= S(x)\oplus \Delta y \}\big \}\Big \},
\end{split}
\end{equation}
where $\triangle x,$\ $\triangle y\in \mathbb{F}_{2}^{n},$ and \textquotedblleft $%
{\normalsize \oplus }$\textquotedblright\ denotes bit-wise addition over $\mathbb{F}_{2}$.

  An S-box $S$ is highly secure against differential attack if its DAP is close to $1/2^n$.
  In Table~\ref{allanalysis}, the $\rm DAP$ of $\sigma(52511, 1, N,$ $ Y, 0)$ and other existing S-boxes is given.
  Note that the DAP of the proposed S-box $\sigma(52511, 1, N, Y, 0)$ is $0.016$ which is close to the optimal value $0.0039$.
 Furthermore, it is evident from Table~\ref{allanalysis} that the DAP of the proposed S-box is less than the S-boxes in~\cite{Asif,ikram, umar2, Azam, YW, GT, Jaki, Ozkaynak2,Bhattacharya,YWang,IChing,AGautam,Shi}, and hence the proposed S-box scheme can generate S-boxes with high resistance against differential attack.
\subsection{Analysis of Boolean Functions}
It is necessary to analyze the boolean functions of a given S-box to measure its confusion/diffusion creation capability.
For an $n \times n$ S-box, strict avalanche criterion SAC$(S)$ and bit independence criterion BIC$(S)$ are used to analyze its boolean functions.
The SAC$(S)$ and the BIC$(S)$ are computed by two matrices $M(S) = [m_{ij}]$ and $B(S)=[b_{ij}]$, respectively, such that
\begin{equation}\label{sAc}
{ m_{ij}=\frac{1}{2^{n}}\left( \sum_{x\in \mathbb{F}_{2}^{n}}{ w}%
\left( {S}_{i}{ (x\oplus \alpha }_{j}{ %
)\oplus S}_{i}{ (x)}\right) \right) { ,}}
\end{equation}
and
\begin{equation}\label{bIc}
\begin{split}
&b_{ij}=\frac{1}{2^{n}}\Biggl(\sum_{\substack{ x\in \mathbb{F}_{2}^{n}\\%
1\leq r\neq i\leq n}}w\biggl(S_{i}(x \oplus \alpha_{j})\oplus S_{i}(x)\oplus \\
&\hspace{4cm}S_{r}(x+\alpha_{j})\oplus S_{r}(x)\biggr) \Biggr),
\end{split}
\end{equation}
where $w(y)$ is the hamming weight of $y$, $\alpha_{j}\in \mathbb{F}_{2}^{n}$
such that $w(\alpha_{j})=1$, $ {S}_{i}$ and ${S}_{r}$ are $i$-th and $r$-th boolean functions of $ S$, respectively, and $1\leq i,j,r \leq n$.
An S-box $S$ satisfies the SAC and the BIC if each non-diagonal entry of $M(S)$ and $B(S)$ have value close to $0.5$.
The maximum and minimum values of the SAC (resp., BIC) of the proposed S-box $\sigma(52511, 1, N, Y, 0)$ are $0.563$ and $0.438$ (resp., $0.521$ and $0.479$).
Note that these values are closed to $0.5$, and hence the proposed S-box satisfies the SAC and the BIC.
Similarly, the SAC and the BIC of some other S-boxes are listed in Table~\ref{allanalysis} and compared with the results of the proposed S-box.
It is evident from Table~\ref{allanalysis} that the proposed S-box can generate more confusion and diffusion as compared to some of the listed S-boxes.
\begin{table}[htb]
\caption{Comparison of the proposed and other existing S-boxes}
\label{allanalysis}
\centering
\resizebox{\columnwidth}{!}
{
\bgroup
\def\arraystretch{1.1}
\begin{tabular}{cccccccccc}
    \hline
    S-boxes&Type&\multicolumn{3}{c}{Linear}& DAP & \multicolumn{4}{c}{Analysis of}\\
           &of&\multicolumn{3}{c}{ Attacks}&  & \multicolumn{4}{c}{Boolean Functions}\\
    \cline{3-5}  \cline{7-10}
    &S-box&NL&LAP&AC & & SAC     &SAC     &BIC     &BIC\\
    &     &  &   &   & & (max)   &(min)   &(max)   &(min)\\
    \cline{1-10}
 Ref.~\cite{Bhattacharya}&&102  &0.133&254  &0.039	&0.562&	0.359  &0.535&0.467  \\
 Ref.~\cite{IChing}&other&108          &0.133&255&0.039    &0.563 &0.493& 0.545 &0.475\\
 Ref.~\cite{AES}&&112  & 0.062 & 09 &0.016  &0.562  &0.453  & 0.504 & 0.480 \\
 Ref.~\cite{Shi}& &108  & 0.156 & 255 &0.046  &0.502  &0.406  & 0.503 & 0.47 \\
 \cline{1-10}
 Ref.~\cite{YW}& &108 & 0.145 &255  & 0.039 &0.578  &0.406  &0.531&0.470  \\
 Ref.~\cite{GT}& &103 &0.132 &255  & 0.039 &0.570  &0.398  &0.535&0.472  \\
 Ref.~\cite{Jaki}&Chaos&100  & 0.129 &255  & 0.039 & 0.594 & 0.422 &0.525&0.477  \\
 Ref.~\cite{Ozkaynak2}&&100 &0.152  &255  &0.039  &0.586  &0.391  &0.537&0.468  \\
 Ref.~\cite{YWang}&&110  & 0.125 & 255 & 0.039 & 0.562 &0.438  &0.555  & 0.473 \\
 Ref.~\cite{AGautam}&& 74 &0.211  &253  &0.055 & 0.688 &0.109  & 0.551 & 0.402 \\
 \cline{1-10}
 Ref.~\cite{Asif}&&104&0.145&255&0.039&0.625&0.391&0.531&0.471  \\
 Ref.~\cite{umar2}&&106&0.148&254&0.039&0.578&0.437&0.535&0.464  \\
 Ref.~\cite{ikram}&EC&106&0.188&253&0.039 &0.609&0.406&0.527&0.465  \\
 Ref.~\cite{Azam}&&106&0.148&255&0.039&0.641&0.406&0.537&0.471  \\
 $\sigma(52511, 1, N, Y, 0)$&&112	&0.063& 255   &0.016&	0.563&	0.438&	0.521& 0.479 \\
    \hline
\end{tabular}
\egroup
}
\end{table}

\subsection{Distinct S-boxes}
 An S-box generator is useful to resist cryptanalysis if it can generate a large number of distinct S-boxes~\cite{Azam}.
For the parameters $p=263, b =1, m=256$ and $k=0$ the number of $(256,263)$-complete S-boxes Num$(E_{263,1}, 256, 0)$ is $128$. 
It turned out with the computational results that all of these $(256,263)$-complete S-boxes are distinct.
However this is not the case in general.

An $(m,p)$-complete S-box $\sigma(p,b,\prec,Y, k)$ is said to be a \textit{natural} $(m,p)$\textit{-complete S-box} if $Y = [0, m-1]$.
For a prime $p$ and ordering $\prec$, let $p^*$ denote the largest integer such that $p^*  \leq p-1$ and there exists at least two ordered MECs $(E_{p,b_1})$ and $E_{p,b_2}$ due to which the natural $(p^*, p)$-complete S-boxes are identical, i.e., for any fixed $m \geq p^*$ the number of natural $(m, p)$-complete S-boxes due to all ordered MECs with prime $p$, ordering $\prec$ and $k = 0$ is equal to $p-1$.
%
A plot of primes $p \in [11, 7817]$ and the integers $p^{*}$ is given in Fig.~\ref{LI}, where the underlying ordering is the natural ordering $N$.
For the orderings $D$ and $M$, such plots are similar as that of $N$.
It is evident from Fig.~\ref{LI}, that with the increase in the value of prime, there is no significant increase in the value of $p^*$ and the largest value of $p^*$ for these primes is $12$.
Hence, for each of these primes, each $m \geq 13$ and $k = 0$, we can get $p-1$ distinct natural $(m, p)$-complete S-boxes with $k = 0$.

\begin{figure}[htb!]
\includegraphics [scale=0.68]{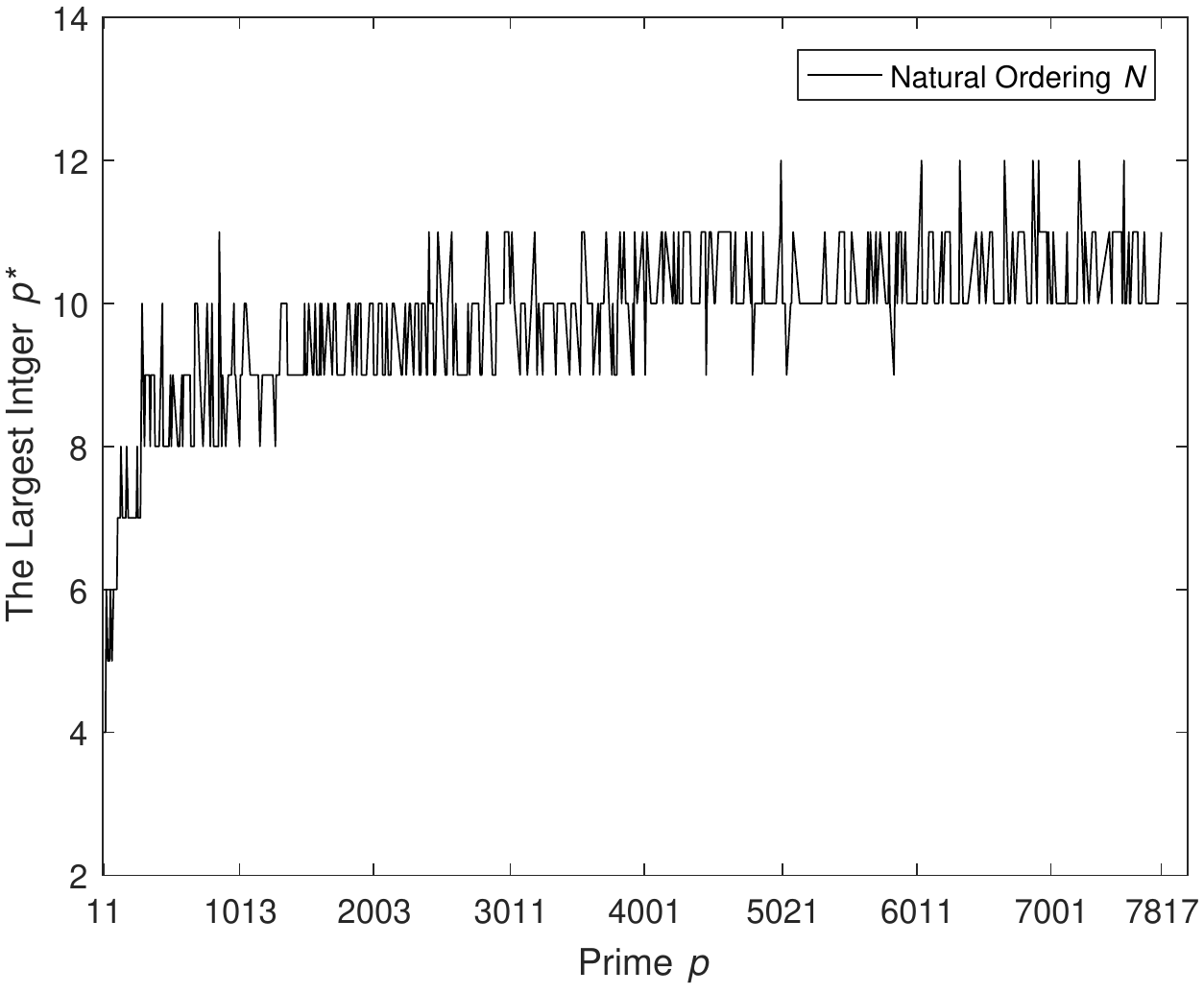}
\caption{Plot of primes $p$ and their corresponding largest integers $p^{*}$}
\label{LI}
\end{figure}
\begin{lemma}\label{class1}
Let $\prec$ be a fixed total order on all MECs in $\mathcal{C}_{1}$ such that for each MEC $E_{p,b} \in \mathcal{C}_1$ it holds that the points $(0, \pm y)$, where $-y$ is additive inverse of $y$ in $\mathbb{F}_p$, have indices from the set $\{1, 2\}$ in the sequence representation of the MEC.
Then for a fixed integer $k \in [0, m-1]$, the number of distinct natural $(m,p)$-complete S-boxes generated by all MECs in $\mathcal{C}_{1}$ are at least
\begin{equation}\label{L10}
\left\{
\begin{tabular}{ll}
$m - 1$ & \textnormal{if} $m < (p-1)/2$\\
$(p-1)/2 $ & \textnormal{otherwise}.

\end{tabular}
\right.
\end{equation}
\end{lemma}
\begin{proof}
Let $E_{p,b}$ be an MEC in $\mathcal{C}_1$, where $b \in [1, p-1]$.
Then by Lemma~\ref{x_zero}, $(0, y) \in E_{p,b}$ for some $y \in [1, p-1]$.
Further by the fact that if $(x, z) \in E_{p,b}$ then $(x, -z) \in E_{p,b}$, where $-z$ is the additive inverse of $z$ in the field $\mathbb{F}_p$, implies that $(0, \pm y) \in E_{p, b}$.
Moreover, by the group theoretic argument exactly one of the integers $y$ and $-y$ belongs to the interval $[0, (p-1)/2]$.
Hence, for a fixed $k \in [0, m-1]$ and the natural $(m,p)$-complete S-box $\sigma(p, b, \prec, Y, k)$ it holds that $\sigma(p, b, \prec, Y, k)(k)\in \{\pm y\}$ if $(0, \pm y)$ have indices from the set $\{1,2\}$ in the sequence representation of $E_{p,b}$.
Note that a point $(0,z)$ cannot appear on two different MECs $E_{p,b_1}$ and $E_{p,b_2}$, otherwise this implies that $b_1 = b_2$.
Thus, for any two MECs $E_{p,b_1},E_{p,b_2}$  in $\mathcal{C}_1$ satisfying the conditions given in the lemma it holds that the natural $(m,p)$-complete S-boxes $\sigma(p, b_1, \prec, Y, k)$ and  $\sigma(p, b_2 , \prec, Y, k)$ have different images at a fixed input $k \in[0, m-1]$.
Thus $|\mathcal{C}_1| = (p-1)/2$ implies the required result.
\end{proof}
For three different primes $p$ distinct S-boxes are generated by the proposed method, and compared with the existing schemes over ECs as shown in Table~\ref{counting}. It is evident that the proposed S-box generator performs better than other schemes.
\begin{table}[htb!]
\caption{Comparison of the number of distinct $8 \times 8$ S-boxes generated by different schemes}
\label{counting}
\centering
\resizebox{\columnwidth}{!}
{
\bgroup
\def\arraystretch{1.27}
\begin{tabular}{lcclll}
\cline{3-6}
&&$p$           & 1889 & 2111 & 2141           \\ 
\cline{3-6}
&&$b$           & 1888 & 1 & 7           \\ \hline
Distinct S-boxes by the& $N$ && 32768$^\dag$ &32768$^\dag$ &32768$^\dag$   \\ \cline{2-6}
proposed method due to& $D$ && 31744$^\dag$&32704$^\dag$&30720$^\dag$    \\ \cline{2-6}
the ordering& $M$ && 15360$^\dag$ & 26748$^\dag$  & 21504$^\dag$   \\ \hline
                       & $N$ && 944 & 1055 & 1070   \\ \cline{2-6}
Distinct S-boxes by Ref.~\cite{ikram}& $D$ && 944 &1055  &1070    \\ \cline{2-6}
                       & $M$ && 944 &1055  &1070    \\ \hline
Distinct S-boxes by Ref.~\cite{umar2}&&  & 50 & 654 & 663  \\ \hline
Distinct S-boxes by Ref.~\cite{Asif}&&  & 1 & 1 & 1  \\ \hline
\end{tabular}
\egroup
}

\end{table}
\blfootnote{The number $h^{\dag}$ stands for an integer greater than $h$.}
\subsection{Fixed Point Test}
An S-box construction scheme is cryptographically good if the average number of fixed points in the constructed S-boxes is as small as possible~\cite{Azam}.  The average number of fixed points of the above generated S-boxes are shown in Table~\ref{fixed}. The experimental results indicate that the proposed S-box generator generates S-boxes with a very small number of fixed points.
Furthermore,  the average number of fixed points in the proposed S-boxes are comparable with that of the existing schemes over ECs.
\begin{table}[htb!]
\caption{Comparison of average number of the fixed points in the S-boxes generated by different schemes}
\label{fixed}
\centering
\resizebox{\columnwidth}{!}
{
\bgroup
\def\arraystretch{1.1}
\begin{tabular}{lcclll}
\cline{3-6}
&&$p$           & 1889 & 2111 & 2141           \\ 
\cline{3-6}
&&$b$           & 1888 & 1 & 7           \\ \hline
Avg.  \# fixed points by the& $N$ && 1.1298 & 1.0844 & 1.0972   \\ \cline{2-6}
proposed method due to& $D$ && 0.9471 & 0.8569  & 0.9393    \\ \cline{2-6}
the ordering& $M$ && 0.8361 & 1.1847  & 1.0025    \\ \hline
                       & $N$ && 1.77 & 0.9735 & 0.9785   \\ \cline{2-6}
Avg.  \# fixed points by Ref.~\cite{ikram}& $D$ && 1.932 &0.9716  &0.9561    \\ \cline{2-6}
                       & $M$ && 1.332 &1.0019  &1.0150    \\ \hline
Avg.  \# fixed points by Ref.~\cite{umar2}&&  & 2.04 & 0.8976 & 0.9351  \\ \hline
Avg.  \# fixed points by Ref.~\cite{Asif}&&  & 2 & 3 & 0  \\ \hline
\end{tabular}
\egroup
}
\end{table}
\subsection{Correlation Test} The correlation test is used to analyze the relationship among the S-boxes generated by any scheme. A robust scheme generates S-boxes with low correlation~\cite{Azam}.
The proposed method is evaluated by determining the correlation coefficients (CCs) of the designed S-boxes. The lower and upper bounds for their CCs are listed in Table~\ref{CC}, which reveal that the proposed scheme is capable of constructing S-boxes with very low correlation as compared to the other schemes over ECs.
\begin{table}[htb!]
\caption{\ Comparison of CCs of S-boxes generated by different schemes}
\label{CC}
\centering
\resizebox{\columnwidth}{!}
 {
{
\bgroup
\def\arraystretch{1.1}
\begin{tabular}{lcccccc}
\hline
Scheme&$p$& $b$ & Ordering &\multicolumn{3}{c}{Correlation}  \\
\cline{5-7}
& & & &Lower & Average & Upper\\
\hline
        &1889 & 1888&$N$ &-0.2685&0.0508 & 0.2753     \\
Proposed&1889 & 1888&$D$ &-0.2263&0.0523 & 0.2986      \\
        &1889 & 1888&$M$ &-0.2817&0.0506 & 0.2902     \\
\hline
        &2111 & 1&$N$ &-0.2718&0.0504 & 0.2600     \\
Proposed&2111 & 1&$D$ &-0.2596&0.0531 & 0.3025      \\
        &2111 & 1&$M$ &-0.2779&0.0507 & 0.2684     \\
\hline
        &2141 & 7&$N$ &-0.2682&0.0503 & 0.2666     \\
Proposed&2141 & 7&$D$ &-0.2565&0.0517 & 0.2890      \\
        &2141 & 7&$M$ &-0.2744&0.0503 & 0.2858     \\
\hline
                 &1889 & 1888&$N$ &-0.2782&0.0503 & 0.2756     \\
Ref.~\cite{ikram}&1889 & 1888&$D$ &-0.4637&-0.0503 & 0.2879      \\
                 &1889 & 1888&$M$ &-0.2694&0.0501 & 0.4844     \\
\hline
                 &2111 & 1&$N$ &-0.2597&0.0504 & 0.2961     \\
Ref.~\cite{ikram}&2111 & 1&$D$ &-0.3679&0.0500 & 0.3996      \\
                 &2111 & 1&$M$ &-0.2720&0.0499 & 0.3019     \\
\hline
                 &2141 & 7&$N$ &-0.2984&0.0500 & 0.3301     \\
Ref.~\cite{ikram}&2141 & 7&$D$ &-0.2661&0.0500 & 0.2639      \\
                 &2141 & 7&$M$ &-0.2977&0.0501 & 0.2975     \\
\hline
Ref.~\cite{umar2} & 1889& 1888 &-- &-0.0025&0.2322&0.9821\\                                                  \hline
Ref.~\cite{umar2} & 2111& 1 &-- &-0.2932&0.0785&0.9988\\
\hline
Ref.~\cite{umar2} & 2141& 7 &-- &-0.2723 &0.0629&0.9999\\
\hline
\end{tabular}
\egroup
}}
\end{table}
\subsection{Time and Space Complexity}
For a good S-box generator it is necessary to have low time and space complexity~\cite{Azam}.
Time and space complexity of the newly proposed method are compared with some of the existing methods in Table~\ref{TandSCom}. It follows that for a fixed prime the proposed method can generate an S-box with low complexity and space as compared to other listed schemes.
This fact makes the proposed S-box generator more efficient and practical.
\begin{table}[htb!]
\caption{Comparison of time and space complexity of different S-box generators over ECs}
\label{TandSCom}
\centering
\resizebox{\columnwidth}{!}
{\
\bgroup
\def\arraystretch{1.2}
{\begin{tabular}{cccccc}
\hline
S-box & Ref.~\cite{Asif} &Ref.~\cite{umar2} &Ref.~\cite{ikram}&\multicolumn{2}{c}{\small Proposed method} \\
\cline{5-6}
& & & &Algorithm~\ref{algo_sbox_1}&Algorithm~\ref{algo_sbox_2}\\
\hline
Time complexity & $\mathcal{O}(p^{2})$ & $\mathcal{O}(p^{2})$ & $\mathcal{O}(mp)$ &  $\mathcal{O}(mp)$  & $\mathcal{O}(m\log m)$ \\
\hline

Space complexity & $\mathcal{O}(p)$ & $\mathcal{O}(p)$ & $\mathcal{O}(m)$ &  $\mathcal{O}(m)$  &  $\mathcal{O}(m) $ \\
 \hline
\end{tabular}%
\egroup
}}
\end{table}
\section{The Proposed Random Numbers Generation Scheme}\label{RNG}
For an ordered MEC $(E_{p,b}, \prec)$, a subset $A \subseteq [0, p-1]$, an integer $m \in [1, |A|]$ and a non-negative integer $k \in [0, m-1]$, we define a sequence of \textit{pseudo random numbers} (SPRNs) $\gamma(p, b, \prec, A, m, k)$ to be a sequence of length $|A|$ whose $i$-th term is defined as $\gamma(p, b, \prec, A, m, k)(i) = y_{(i+k)\pmod m}\pmod{m}$, where $y_{(i+k)\pmod m}$ is the $(i+k)\pmod m$-th element of the ordered set $(A, \prec^*)$ in its sequence representation.\\
One of the differences in the definition of an $(m,p)$-complete S-box and the proposed SPRNs is that an $(m,p)$-complete set is required as an input for the S-box generation, since an S-box of length $m$ is a permutation on the set $[0, m-1]$.
Furthermore, Algorithm~\ref{algo_sbox_1} and~\ref{algo_sbox_2} can be used for the generation of the proposed SPRNs, however, we propose an other algorithm which is more efficient than Algorithm~\ref{algo_sbox_2} for its generation.
This new algorithm is also based on Observation~\ref{obs_set}, but there is no constraint on $A$ to be an $(m,p)$-complete set, and hence we can generate all proposed SPRNs for a given prime $p$ by using $E_{p, R(p, \mathcal{C}_i)}$, where $i \in \{1, 2\}$.

\begin{algorithm}[H]
	\caption{{\bf Constructing the proposed SPRNs using EC isomorphism}}\label{algo_sprn}
\begin{algorithmic}[1]
\REQUIRE An MEC $E_{p, R(p,\mathcal{C}_{i})}$, where $i\in \{1, 2\}$, an integer $t \in [1, (p-1)/2]$, a total order $\prec$ on the MEC $E_{p,t^{6}R(p,\mathcal{C}_{i})}$ and a subset $Y \subseteq [0, p-1]$.
\ENSURE The proposed SPRNs $\gamma(p, t^{6}R(p,\mathcal{C}_{i}), \prec, t^3Y, m, k)$.

\STATE $A:=\emptyset$; /*A set containing the points of $E_{p,t^6R(p,\mathcal{C}_{i})}$ with $y$-coordinates from the set $t^3Y$*/

\FORALL {$y \in Y$}
\STATE Find $x \in [0, p-1]$ such that $(x, y) \in E_{p,R(p,\mathcal{C}_{i})}$;
\STATE $A:= A \cup \{(t^2x, t^3y)\}$;
\ENDFOR;
\STATE Sort $A$ w.r.t. the element of the total order $\prec^*$;
\STATE $\pi := (\pi(0), \pi(1), \ldots, \pi(|A|-1))$;
\FORALL {integer $i \in [0, |A|-1]$}
	\STATE $\pi(i) := a_{(i+k)\Mod{m}}\pmod{m}$, where $a_{(i +k)\pmod{m}}$ is the $(i +k)\pmod{m}$-th element of the ordered set $(A, \prec^*)$
\ENDFOR;
\STATE Output $\pi$ as the proposed SPRN $\gamma(p, t^{6}R(p,\mathcal{C}_{i}), \prec, t^3Y, m, k)$
	
\end{algorithmic}
\end{algorithm}
\noindent
Note that the time and space complexity of Algorithm~\ref{algo_sprn} are $\mathcal{O}(|A|\log {|A}|$) and $\mathcal{O}(|A|)$ respectively as obtained for Algorithm~\ref{algo_sbox_2}.
However, Algorithm~\ref{algo_sprn} does not require $t^{-1}$ as an input parameter to compute $\gamma(p, t^{6}R(p,\mathcal{C}_{i}), \prec,$ $t^3Y, m, k)$ for which we need preprocessing.
Furthermore, Lemma~\ref{class1} trivially holds for our proposed SPRNs.
This implies that the proposed PRNG can generate a large number of distinct SPRNs for a given prime.

\section{Analysis of the Proposed SPRNs Method}\label{RA}
We applied some well-known tests to analyze the strength of our proposed SPRNs. A brief introduction to these tests and their experimental results are given below.
We used orderings $N, D$ and $M$ for these tests.
\subsection{Histogram and Entropy Test}
Histogram and entropy are the two widely used tests to measure the extent of randomness of a RNG.
For a sequence $X$ over the set of symbols $\Omega$, the histogram of $X$ is a function $f_{X}$ over $\Omega$ such that for each $w \in \Omega$, $f_{X}(w)$ is equal to the number of occurrences of $w$ in $X$.
We call $f_X(w)$, the frequency of $w$ in $X$.
A sequence $X$ has uniform histogram if all elements of its symbol set have same frequency.
The histogram test is a generalization of the Monobit test included in NIST STS~\cite{Rukhin}.
A sequence is said to be highly random if it has a uniformly distributed histogram.

Shannon~\cite{Shannon} introduced the concept of entropy.
For a sequence $X$ over the set of symbols $\Omega$,
the entropy H$(X)$ of $X$ is defined as
\begin{equation}\label{formula}
{\rm H}(X)= -\sum_{w \in \Omega}\frac{f_X(w)}{|X|}\log_{2}(\frac{f_X(w)}{|X|}).
\end{equation}
The upper bound for the entropy is log$_{2}(|\Omega|).$
The higher is the value of entropy of a sequence the higher is the randomness in the sequence.
\begin{remark}
 For any distinct $k_1, k_2 \in [0, m-~1]$, the histograms of the proposed SPRNs $\gamma(p,b,\prec,A,m, k_1)$ and $\gamma(p,b,\prec,A,m,k_2)$ are the same, and hence ${\rm H}(\gamma(p,b,\prec,$ $A,m,k_1)) = {\rm H}(\gamma(p,b,\prec,A,m,k_2))$.
 \end{remark}
In the next lemmas we discuss when the proposed SPRNs has uniformly distributed histogram and its entropy approaches the optimal value.
\begin{lemma}\label{opt}
For an $(m,p)$-complete set $A$, a positive integer $h \leq m$ such that $m = hq + r$, a non-negative integer $k \leq h$, and the SPRNs $X = \gamma(p,b,\prec,A,h,k)$ it holds that
\begin{enumerate}[label=\roman*, font=\upshape, noitemsep]
\item[(i)] \begin{equation*}
 f_X(w) = \left\{
 \begin{tabular}{ll}
$q+1$ & \textnormal{if} $w \in [0, r-1]$,\\
$q$  & \textnormal{otherwise},
 \end{tabular}
 \right.
 \end{equation*}
 \textnormal{if}  $r \neq 0$ \textnormal{and} $A = [0, m-1]$,
 \item[(ii)] \textnormal{for each} $w \in [0, h-1]$, $f_X(w) = q$ \textnormal{if}  $r = 0$.
\end{enumerate}


\end{lemma}
\begin{proof}
It is trivial that the domain of the histogram of $X$ is the set $[0, h-1]$.\\
(i)~If $r \neq 0$ and $A = [0, m-1]$, then it can be easily verified that $A$ can be partitioned in $q+1$ sets $\{ih +\ell \mid 0 \leq \ell \leq h-1\}$, where $0 \leq i \leq q-1$, and $\{qh +\ell \mid 0 \leq \ell \leq r -1\}$.
This implies that for each $w \in [0, h-1]$ it holds that
\begin{equation*}
 f_X(w) = \left\{
 \begin{tabular}{ll}
$q+1$ & \textnormal{if} $w \in [0, r-1]$,\\
$q$  &\textnormal{otherwise}.
 \end{tabular}
 \right.
 \end{equation*}
(ii)~If $r = 0$, then $m = hq$.
 We know that for each $a \in A$, it holds that $a = mi + j$, where $ 0 \leq j \leq m- 1$.
Thus, with the fact that $m = hq$, it holds that
\begin{equation*}
\begin{split}
a \pmod h &= ((mi) \pmod h + j \pmod h) \\
 &= j \pmod h
 \end{split}
\end{equation*}
This implies that $\{a\pmod{m} \mid a \in A\}$ = $\{(a\pmod{m}) \pmod {h} \mid a \in A\}$.
Thus by using the same reason, we can partition $A$ into $q$ sets, since $m = hq$, and hence $f_{X}(w) = q$ for each $w \in [0, h-1]$.
This completes the proof.
\end{proof}
For the parameters given in Lemma~\ref{opt}, we can deduce that the histogram of our proposed SPRNs is either approximately uniform or exactly uniform.

\begin{corollary}
Let $A$ be an $(m,p)$-complete set, $h \leq m$ such that $m = hq + r$ be a positive integer, $k \leq m-1$ be a non-negative integer, and $X$ be the proposed SPRNs $ \gamma(p,b,\prec,$ $A,h,k)$.
It holds that

 \begin{equation}
 \rm {H}(\it {X}) = \begin{cases}
        -r(\frac{q+1}{|X|})\rm{log_{2}}\it(\frac{q+1}{|X|})- & \textnormal {if}\quad r \neq 0, A = [0, m-1],\\
        (h-r)(\frac{q}{|X|})\rm{log_{2}}\it(\frac{q}{|X|}) &\\
        \log_2(h) & \textnormal{if}\quad r = 0.
        \end{cases}
  \end{equation}
\end{corollary}
 \begin{proof}
When $r \neq 0$ and $A = [0, m-1]$, then by Lemma~\ref{opt}~(i), there are $r$ (resp., $h-r$) numbers in $[0, h-1]$ whose frequency is $q+1$ (resp., $q$), and therefore we have the result.\\
When $r = 0$, then by Lemma~\ref{opt}~(ii), all numbers in $[0, h-1]$ have frequency $q$ and there are $h$ elements in $[0, h-1]$ and hence the result.
\end{proof}

To test the efficiency of the proposed PRNG, we generated SPRNs $X_1 = \gamma(52511, 1, N, A, 127, 0)$, $X_2 = \gamma(52511, 1, N, A, 16, 0)$, where $A$ is the set given in Table~\ref{FT}, $X_3 = \gamma(101, 35, N, [0, 100], 6, 0)$ and $X_4 = \gamma(3917, 301, N, [0, 3916], 3917, 0)$.
The histogram of $X_1$ is given in Fig.~\ref{Histo1} which is approximately uniform, while by Lemma~\ref{opt}
the histograms of $X_3$ and $X_4$ are uniformly distributed.
Furthermore, the entropy of each of these SPRNs is listed in Table~\ref{tests}.
Observe that the newly generated SPRNs have entropy close to the optimal value.
Thus, by histogram and entropy test it is evident that the proposed method can generate highly random SPRNs.
Moreover, the proposed SPRNs $X_4$ are compared with the SPRNs $\mathcal{R}(3917,0,301,10,2)$ generated by the existing technique due to Hayat~\cite{umar2} over ECs.
By Lemma~\ref{opt} it holds that  $f_{X_{4}}(w) = 1$ for each $w \in [0, 3916]$, and by Fig.~\ref{Histo2}, it is clear that the histogram of $X_4$ is more uniform as compared to that of the SPRNs $\mathcal{R}(3917,0,301,10,2)$.
By Table~\ref{tests}, the entropy of $X_4$ is also higher than that of $\mathcal{R}(3917,0,301,10,2)$, and hence the proposed PRNG is better than the generator due to Hayat~\cite{umar2}.
\begin{table}[htb!]
\caption{Comparison of entropy and period of different sequence of random numbers over EC}
\label{tests}
\centering
\resizebox{\columnwidth}{!}
 {
\bgroup
\def\arraystretch{1.1}
\begin{tabular}{llllllll}
\hline
\multicolumn{1}{c}{Random sequence $X$} & Type of A   & H$(X)$ & log$_{2}(|\Omega|)$  & Period & Optimal\\
&     &     & & &period \\
\hline
$\gamma(52511,1,N,A,127,0)$&Table~\ref{FT} 	&6.6076  &6.7814&256 &256  \\
$\gamma(52511,1,N,A,16,0)$&Table~\ref{FT} &4 &4&256 &256 \\
$\gamma(101,35,N,A,6,0)$&$[0, 100]$   & 2.5846 &2.5850& 99&101\\
$\gamma(3917, 301, N, A, 3917, 0)$ & [0, 3916] & 11.9355 &11.9355&3917 &3917\\
$\mathcal{R}(3917,0,301,10,2)$ Ref.~\cite{umar2}&\multicolumn{1}{c}{--} & 10.9465&11.1536&3917&3917\\
\hline
\end{tabular}
\egroup
}
\end{table}
\begin{figure}[htb!]
\includegraphics [scale=0.65]{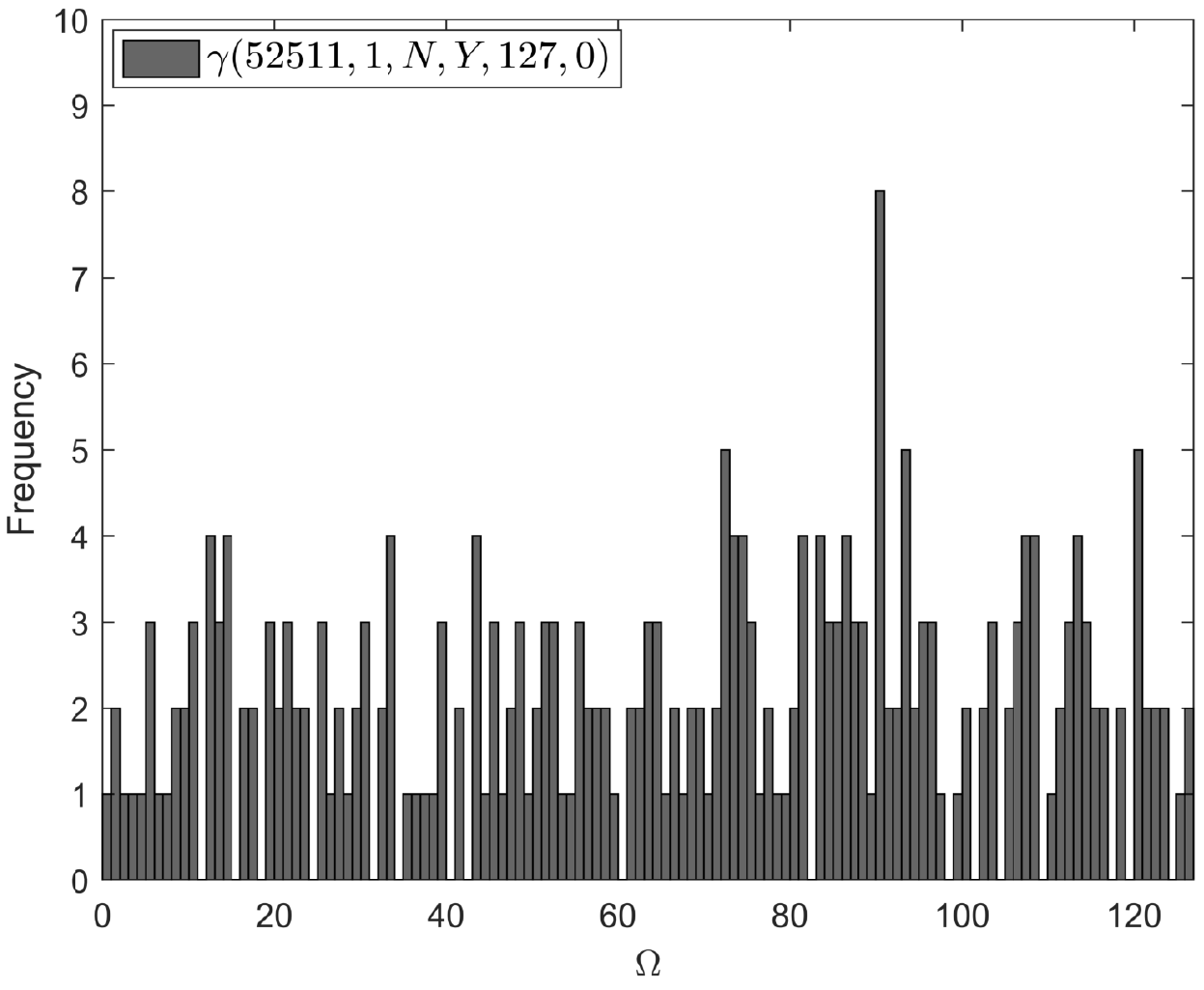}
\caption{The histogram of $\gamma(52511,1,N,Y,127,0)$}
\label{Histo1}
\end{figure}

\begin{figure}[htb!]
\includegraphics [scale=0.65]{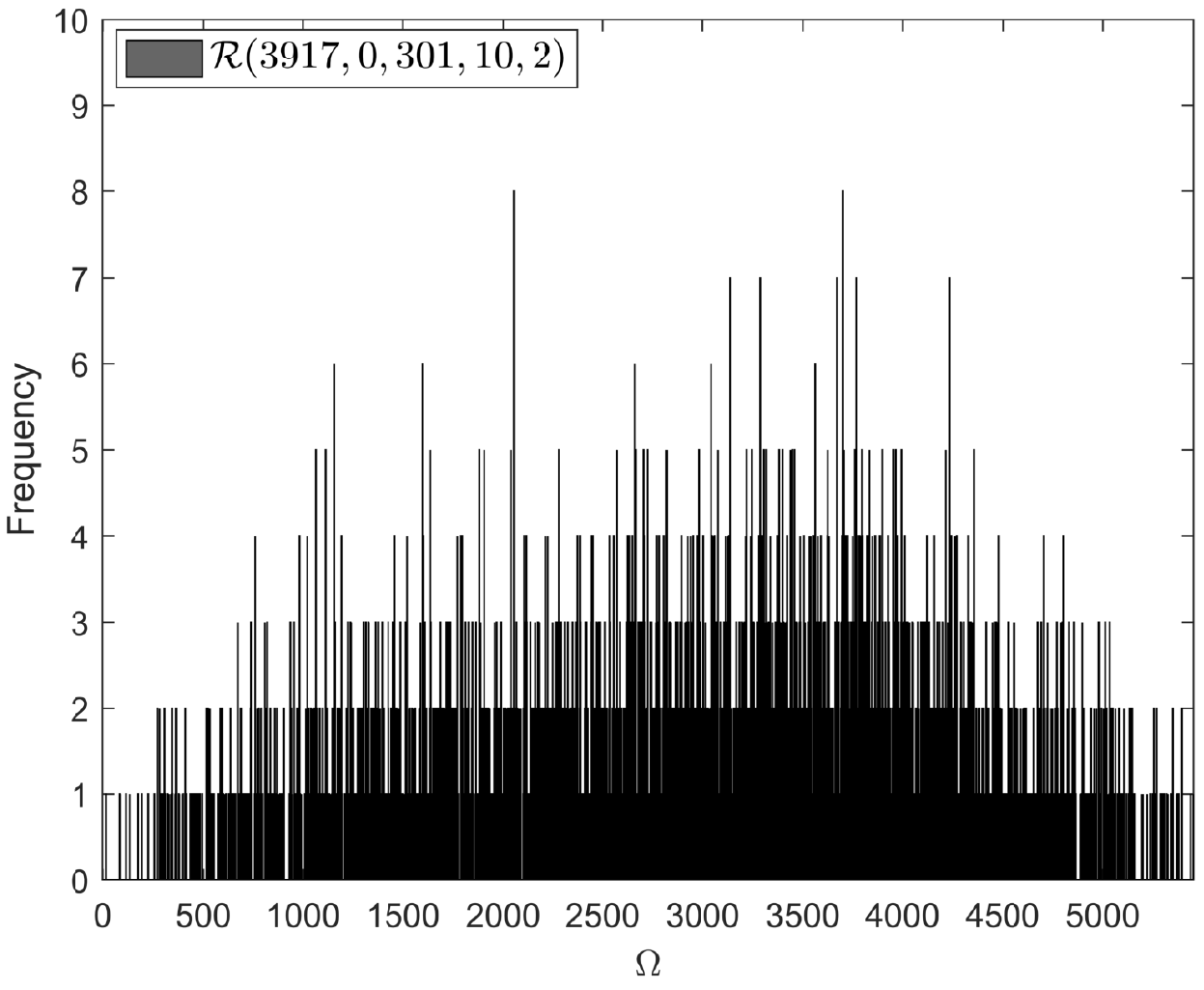}
\caption{The histogram of $\mathcal{R}(3917,0,301,10,2)$}
\label{Histo2}
\end{figure}
\subsection{Period Test}
Period test is another important test to analyze the randomness of a PRNG.
A sequence $X = \{a_{n}\}$ is said to be periodic if it repeats itself after a fixed number of integers, i.e., $\{a_{n+h}\}=\{a_{n}\}$ for the least positive integer $h$.
In this case $h$ is called the period of the sequence $X$.
The maximum period that a sequence $X$ can have is $|X|$.
The sequence $X$ is said to be highly random if its period is long enough~\cite{Marsaglia}. 
We computed the period of the proposed SPRNs $X_i, i = 1, 2, \ldots, 4$ and the SPRNs $R(3917,0,301,10,2)$ generated by the scheme proposed in \cite{umar2} and the results are listed in Table~\ref{tests}.
It is evident from Table~\ref{tests} that the proposed SPRNs have period colse to the optimal value.
Hence, the proposed PRNG can generated highly random numbers.



\subsection{Time and Space Complexity}
It is necessary for a good PRNG to have low time and space complexity.
The time and space complexity of the proposed PRNG and the generator proposed by Hayat et al.~\cite{umar2} are compared in Table~\ref{RNcom}.
Note that the time and space complexity of the proposed PRNG depend on the size of the input set, while the time and space complexity of PRNG due to Hayat et al.~\cite{umar2} are $\mathcal{O}(p^2)$ and $\mathcal{O}(p)$, respectively, where $p$ is underlying prime.
Hence, the proposed PRNG is more efficient as compared to the PRNG due to Hayat et al.~\cite{umar2}.
\begin{table}[H]
\caption{Comparison of time and space complexity of different PRNGs over ECs}
\label{RNcom}
\centering
\resizebox{\columnwidth}{!}
{\
\bgroup
\def\arraystretch{1.2}
{\begin{tabular}{cccc}
\hline
 & Input size $m$& Ref.~\cite{umar2} & Proposed method\\
 \hline
 Time complexity &\begin{tabular}{@{}l@{}}  $m < p$ \\
 													$m = p$  \end{tabular} & $\mathcal{O}(p^2)$ &$\mathcal{O}(|A|\log |A|)$\\
 \hline
 Space complexity &\begin{tabular}{@{}l@{}}  $m < p$ \\
 													$m = p$  \end{tabular} & $\mathcal{O}(p)$ &$\mathcal{O}(|A|)$\\
 \hline

\end{tabular}%
\egroup
}}
\end{table}

\section{Conclusion}\label{Con}
Novel S-box generator and PRNG are presented based on a special class of the ordered MECs.
Furthermore, efficient algorithms are also presented to implement the proposed generators.
The security strength of these generators is tested by applying several well-known security tests.
Experimental results and comparison reveal that the proposed generators are capable of generating highly secure S-boxes and PRNs as compared to some of the exiting commonly used cryptosystems in low time and space complexity.

\end{document}